\documentclass[12pt]{article}
\usepackage{amsthm,amsfonts,amsmath, amscd}

                                \usepackage{verbatim}
\swapnumbers
                              
\theoremstyle{plain}

\numberwithin{equation}{section}

\newcommand{\bR}{{\mathbb R}}

\newcommand{\bC}{{\mathbb C}}

\newcommand{\cH}{{\mathcal H}}
\newcommand{\cC}{{\mathcal C}}

\newcommand{\cS}{{\mathcal S}}

\newtheorem{lemma}{Lemma}
\newtheorem{proposition}{Proposition}
\newtheorem{theorem}{Theorem}

\newcommand{\ket}[1]{\left\vert #1\right\rangle}
\newcommand{\bra}[1]{\left\langle #1\right\vert}

\newcommand{\Tr}{\mbox{Tr}}

\newcommand{\sign}{\mathop{\rm sign}}

\newcommand{\be}{\begin{equation}}
\newcommand{\ee}{\end{equation}}
\newcommand{\bea}{\begin{eqnarray}}
\newcommand{\eea}{\end{eqnarray}}
\newcommand{\beann}{\begin{eqnarray*}}
\newcommand{\eeann}{\end{eqnarray*}}

\usepackage{color}

\begin{document}
\title{Complete criterion for convex-Gaussian-state detection}
\author{Anna Vershynina\\
\textit{\small{Institute for Quantum Information, RWTH Aachen University,}}\\
\textit{\small{52056 Aachen, Germany}}\\
\small{annavershynina@gmail.com}}

\date{\today}


\maketitle
\begin{abstract}
We present a new criterion that determines whether a fermionic state is a convex combination of pure Gaussian states. This criterion is complete and characterizes the set of convex-Gaussian states from the inside. If a state passes a program it is a convex-Gaussian state and any convex-Gaussian state can be approximated with arbitrary precision by states passing the criterion.  The criterion is presented in the form of a sequence of solvable semidefinite programs. It is also complementary to the one developed by de Melo, \'{C}wikli\'{n}ski and Terhal, which aims at characterizing the set of convex-Gaussian states from the outside. Here we present an explicit proof that criterion by de Melo et al. is complete, by estimating a distance between an $n$-extendible state, a state that passes the criterion, to the set of convex-Gaussian states. 
\end{abstract}

\section{Introduction}

One of the goals in quantum information theory is to achieve the universal quantum computation. Majorana fermions, first introduced in \cite{Majorana}, give a possible way to realize it. Such Majorana fermions have been reportedly observed in superconductor-semiconductor systems \cite{Mourik}-\cite{Rokhinson}. Braiding (interchanging) of Majorana fermions can be used to implement fault-tolerant gates \cite{Bravyi06}-\cite{Beenakker} and an efficient protocol for braiding Majorana fermions in atomic wire networks that is robust against experimentally relevant errors has been proposed in \cite{Kraus}. A set of universal quantum gates has been proposed in \cite{Bravyi02}, although the gates include a quadric interaction between Majorana fermions. In \cite{Bravyi06} it was shown that any computation that uses only braiding operators on a system of Majorana fermions initially prepared in a fermionic Gaussian state can be efficiently simulated classically. It was also shown there that the universal computation can be carried out if a supply of  two specific states, $\ket{a_4}$ and $\ket{a_8}$ on 4 and 8 Majorana fermions respectively, is available.

In \cite{Terhal} it was shown that if the noisy ancillae are convex mixtures of Gaussian fermionic states and the computation involves only non-interacting fermions operations, one can classically simulate such quantum computation. Hence rises the interest in developing a complete criterion that determines whether a states is a convex mixture of Gaussian fermionic states. A complete criterion is an algorithm that characterizes the target set with an arbitrary precision. 

Several works have been done to find such criteria. In \cite{Terhal} de Melo et al. found a criterion similar to the one established for separable states \cite{Doherty}. Both criteria are based on the notion of 'extendibility' and approach the target set (convex-Gaussian or separable states) from the outside. The criterion in \cite{Terhal} is presented in the form of a sequence of solvable semi-definite programs. The authors in \cite{Oszmaniec} presented another characterization of the set of convex-Gaussian states in a particular case when a state space is on 4 fermionic modes, i.e. involves 8 Majorana fermions.

In this paper we show that the criteria introduced in \cite{Terhal} is complete, a result which was implied there. We prove it by estimating a distance from a $n$-extendible state to the set of convex-Gaussian states and show that this distance converges to zero as $n$ goes to infinity. This shows that any non-convex-Gaussian state will fail the program at some point.

A different criterion for detecting separable states was established in \cite{Navascues} that describes the set of separable states from the inside. Taking inspiration from this work, we establish a criterion characterizing the set of convex-Gaussian states from the inside as well. The criterion is given in the form of a hierarchy of semi-definite programs. A state that passes a program is convex-Gaussian and any convex-Gaussian state can be approximated with an arbitrary precision by states that pass the program. 

The paper is organized as follows. In "Set up" section we establish the basic knowledge on Gaussian and convex-Gaussian states, introduce the necessarily notations and definitions, and present the known results that will be useful later. In "Quantitative bound" section we present and prove the quantitative bound that shows that the criteria established in \cite{Terhal} is complete. In "Criterion for convex-Gaussian states detection" section we develop a criteria characterizing the set of convex-Gaussian states from inside in the form of a hierarchy of semi-definite programs. In Appendix A we provide a complete proof started in \cite{Terhal} of the fact that the null-space of the operator $\Lambda$ (\ref{Lambda}) is a Gaussian-symmetric subspace of a symmetric space. In Appendix B we discuss an equivalent construction of the extendible states and the semi-definite program using the isomorphism between the tensor product of $n$ algebras of $2m$ Majorana fermions $\cC_{2m}^{\otimes n}$  and an algebra with $2mn$ Majorana fermions $\cC_{2mn}$.

\section{Set up}

We consider a system of $m$ fermionic modes, the Hilbert space of which is a Fock space $\text{Fock}(\bC^m)$. Creation and annihialation operators on this space, denoted by $a^\dagger_j$, $a_j$, $j=1,..,m$,  satisfy canonical-anticommutation relations (CAR) $$\{a_j, a^\dagger_k\}=\delta_{jk}I\ \text{ and }\ \{a_j, a_k\}=0.$$ 
The Fock basis is defined by $\ket{n_1,..n_m}=(a^\dagger_1)^{n_1}\cdots (a^\dagger_m)^{n_m} \ket{0}$, where the occupation numbers take values $n_j\in\{0,1\}$ and the vacuum state $\ket{0}$ satisfies $a_j\ket{0}=0$ for all $j=1,...,m$. 

Instead of $2m$ creation and annihilation operators define $2m$ Hermitian operators, called \textit{Majorana fermion operators} \cite{Majorana}:
$$c_{2k-1}=a_k+a_k^\dagger\ \text{and }\ c_{2k}=i(a_k^\dagger-a_k), \text{ }\ k=1,...,m.$$
Majorana operators form a Clifford algebra, which will be denoted as $\cC_{2m}$, i.e. they satisfy the following relations $$c_j=c_j^\dagger,\ \text{and} \ c_j^2=I, \ \text{for } j=1,...,2m;$$ $$c_kc_j=-c_jc_k, \ \text{for } k,j=1,...,2m.$$

Any operator $X\in\cC_{2m}$ can be represented as a polynomial in $\{c_j\}_j$ and identity.  The operator is called \textit{even (odd)} if it can be written as a linear combination of products of an even (odd) number of Majorana operators. In other words, any even operator $X\in\cC_{2m}$ can be written as
\begin{equation}\label{even_op}
X=\alpha_0 I/2^m+\sum_{k=1}^{m}\sum_{1\leq a_1<...<a_{2k}\leq 2m}\alpha_{a_1...a_{2k}}\,i^kc_{a_1}...c_{a_{2k}},
\end{equation}
where the coefficients $\alpha_0=2^{-m}\Tr(X)$ and $\alpha_{a_1...a_{2k}}$ are real. From the properties of Majorana operators each correlator is Hermitian $(i^kc_{a_1}...c_{a_{2k}})^\dagger=i^kc_{a_1}...c_{a_{2k}}.$

Writing an even density matrix $\rho$ in the form (\ref{even_op}), will give us the following restrictions on the coefficients: $\alpha_0=1,$ and for any subset $\{a_1,...,a_{2k}\}\subset\{1,...,2m\}$ we have $|\alpha_{a_1,..,a_{2k}}|\leq1/2^{m-1}.$  These conditions come from the fact that $\rho\geq 0$, $\Tr\rho=1$ and $\|\rho-I/2^m\|_1\leq 2$, (for more insight see the proof of Theorem \ref{thm:inside}). Here the trace norm is defined for any Hermitian operator $M$ as:
\begin{equation}\label{trace_norm}
\|M\|_1:=\Tr\{\sqrt{M^\dagger M}\}=\max_{-I\leq \Pi\leq I}\Tr\{\Pi\, M\}.
\end{equation}

\textit{Fermionic Gaussian state} $\rho$ is an even state of the form $\rho=K\, \exp(-i\sum_{j\neq k}\beta_{jk}c_jc_k)$ with real-antisymmetric matrix $\beta=\{\beta_{jk}\}_{j,k}$ and normalization $K$. Block-diagonalizing matrix $\beta$ with an orthogonal matrix $R\in SO(2m)$, one can re-express any Gaussian state $\rho$ in standard form as
\begin{equation}\label{gaussian_standard}
\rho=\frac{1}{2^m}\prod_{k=1}^m(I+i\lambda_k\tilde{c}_{2k-1}\tilde{c}_{2k}),
\end{equation}
where $\tilde{c}=R^Tc$ and the coefficients $\lambda_j\in [-1, 1]$. Gaussian pure states are the one such that $\lambda_j\in\{-1,1\}$, which can be easily checked from the fact that $\rho^2=\rho$ for pure states.

For any state $\rho\in\cC_{2m}$ define \textit{a correlation matrix} $M$ as a $2m\times 2m$ real anti-symmetric matrix with elements
$$M_{jk}=\frac{i}{2}\Tr(\rho[c_j,c_k]), \ \ \text{for }j,k=1,...,2m.$$
For $j=k$, we have $M_{jj}=0$ and for $j\neq k$, we have $M_{jk}=\frac{i}{2}\Tr(\rho c_jc_k-\rho c_k c_j)=i\Tr(\rho c_jc_k).$  Matrix $M$ is anti-symmetric, because $M_{jk}=-M_{kj}$, for $j\neq k$.

Since any real antisymmetric matrix can be block-diagonalized by an orthogonal matrix $R\in SO(2m)$, the correlation matrix $M$ can be written in the form
$$M=R\bigoplus_{j=1}^m\left(\begin{array}{cc}
0&\lambda_j\\
-\lambda_j& 0\\
\end{array}\right)R^T.$$
For any even state the \textit{singular values} $\{\lambda_k\}_{k=1}^{m}$ of $M$ lie in the interval $[-1,1]$, since each of them is the expectation value of the Hermitian operator $i\tilde{c}_{2k-1}\tilde{c}_{2k}$, which has $\pm 1$ eigenvalues. This condition corresponds to having $M^TM\leq I$ for any even state. The equality $M^TM=I$, or in other words the equality $\lambda_k=\pm 1$, is satisfied if and only if the state is pure Gaussian \cite{Bravyi}.

The correlation matrix $M$ fully determines a Gaussian state, which can be seen from the standard form (\ref{gaussian_standard}). All higher order correlators are determined by Wick's theorem \cite{Wick}.

\textit{A convex-Gaussian} state is defined as a convex combination of pure Gaussian states, in other words, the set of convex-Gaussian states is defined as
\begin{align*}
G_c=\{\rho=\sum_i p_i\ket{\psi_i}\bra{\psi_i}\ : \ &\ket{\psi_i}\text{ is Gaussian } \forall i,\\
&\text{ and }\sum_ip_i=1, p_i\geq 0\}.
\end{align*}

Every Gaussian state is convex-Gaussian, which can be seen from (\ref{gaussian_standard}) by writing each term in the product as $I+i\lambda_k c_{2k-1}c_{2k}=p_k(I+ic_{2k-1}c_{2k})+(1-p_k)(I-ic_{2k-1}c_{2k})$, with $\lambda_k=2p_k-1$. On the other hand, a convex combination of two Gaussian states is convex-Gaussian, but may not necessarily be a Gaussian state, since each of them may require a different orthogonal matrix $R$ for diagonalization (\ref{gaussian_standard}). In other words, the set of Gaussian states is not convex and it lies in a convex set of convex-Gaussian states, which also includes some non-Gaussian states.

The following proposition shows that if a state is close enough to the maximally mixed Gaussian state $I/2^m$, it is convex-Gaussian.

\begin{proposition}\label{Convex}
(\cite{Terhal}, Theorem 1) For any even state $\rho\in\cC_{2m}$ there exists $\epsilon>0$ such that $\rho_\epsilon=\epsilon\rho+(1-\epsilon) I/2^m\in G$ is convex-Gaussian.
\end{proposition}

In other words, a small perturbation of the normalized identity state $I/2^m$ by any state still lie in the set of convex-Gaussian states. From the proof of Proposition \ref{Convex} (\cite{Terhal}) and the proof of Theorem \ref{thm:inside} that we provide later, one can lower bound $\epsilon$ as a function of $m$. We can say that for any state $\rho\in\cC_{2m}$ there exists $\epsilon\geq\frac{1}{1+(2m)!/m!}$ such that the state $\rho_\epsilon=\epsilon\rho+(1-\epsilon) I/2^m\in G_c$ is convex-Gaussian.

Consider the Hermitian operator
\begin{equation}\label{Lambda}
\Lambda=\sum_{j=i}^{2m}c_j\otimes c_j,
\end{equation}
 which belongs to the tensor product of two Clifford algebras $\cC_{2m}\otimes\cC_{2m}$.
It was first introduced in \cite{Bravyi} and was later investigated in \cite{Terhal}. It was shown there that an even state $\rho\in\cC_{2m}$ satisfies 
 \begin{equation*}
 \Lambda(\rho\otimes\rho)=0\text{ }\iff\rho\text{ is a pure Gaussian state.}
 \end{equation*}

The approximate result is also true, i.e. if a state such that $\|\Lambda(\rho\otimes\rho)\Lambda\|_1\leq\epsilon$, it is close to a pure Gaussian state.

\begin{lemma}\label{lemma:app}
If a state $\tau\in\cC_{2m}$ is such that 
$$\|\Lambda(\tau\otimes\tau)\Lambda\|_1\leq \nu, $$
then there exists a Gaussian state $\ket{\psi}$ close to $\tau$, i.e. such that 
$$\|\tau-\ket{\psi}\bra{\psi} \|_1\leq \sqrt{m\,\nu}.$$
\end{lemma}
\begin{proof}
Let $M_{\tau}$ be a correlation matrix of state $\tau$. Write the matrix in a block-diagonal form $$M_\tau=R\bigoplus_{j=1}^m \left( \begin{array}{cc}
0 & \lambda_j \\
-\lambda_j & 0
 \end{array} \right)R^T. $$
 
 Apply a dephasing procedure to $\tau$, that is described in the following way, define $$\tau_0=\tau \text{ and }\tau_k=\frac{1}{2}(\tau_{k-1}+U\tau_{k-1}U^\dagger), \text{ for }k=1,...,m$$ where $U$ is the following FLO, for each $k=1,..,m$
 $$U_kc_{2k}U_k^\dagger=-c_{2k}, \ U_kc_{2k-1}U_k^\dagger=-c_{2k-1}, $$
  and $U_kc_jU^\dagger=c_j$ for $j\neq 2k, 2k-1$, $j=1,..,m$. 

After the dephasing procedure, the state $\tau_m$ contains only mutually commuting operators $ic_{2k-1}c_{2k}$, $k=1,..,m$ and so its eigendecomposition involves the eigenstates of these operators, which are Gaussian pure states. Hence the state $\tau_m$ is convex-Gaussian, i.e. it is in the form $$\tau_m=:\sum_k\alpha_k\ket{\phi_k}\bra{\phi_k},$$ where $\sum_k\alpha_k=1,$ $\alpha_k\geq 0$ and $\ket{\phi_k}\bra{\phi_k}=\frac{1}{2^m}\prod_{j=1}^m(I+i\beta_j^kc_{2j-1}c_{2j})$ is a pure Gaussian state with $\beta_j^k\in\{-1,1\}$ for all $k, j=1,..,m$. Note that each $\ket{\phi_k}$ is an eigenvector to all $ic_{2l-1}c_{2l}$, $l=1,..,m$.
Such dephasing procedure leaves the correlation matrix $M_\tau$ invariant (\cite{Terhal}).

Construct a pure Gaussian state in the following way, let
$$\ket\psi\bra\psi=\frac{1}{2^m}\prod_{j=1}^m(I+i\gamma_jc_{2j-1}c_{2j})$$
where  $\gamma_j=\sign\lambda_j$. 
 
The state $\ket\psi$ is an eigenvector to all $ic_{2k-1}c_{2k}$, because $ic_{2k-1}c_{2k}\ket\psi\bra\psi=\gamma_k\ket\psi\bra\psi.$ Therefore the state $\ket{\psi}$ is contained in the eigendecomposition of a dephased state $\tau_m$. Let it be the first state, i.e. $\ket{\phi_0}\bra{\phi_0}:=\ket\psi\bra\psi$ and let us denote $\beta_j^0:=\gamma_j=\sign\lambda_j$.

Having in mind the relationship between the fidelity and the trace-norm $$\|\tau-\ket{\psi}\bra{\psi}\|_1\leq 2\sqrt{1-F(\tau,\ket{\psi})},$$ consider the fidelity $F(\tau, \ket{\psi})=\bra{\psi}\tau\ket{\psi}.$
Since the dephasing procedure left the correlation matrix $M_\tau$ invariant, the fidelity $F(\tau, \ket{\psi})$ is the same before and after the dephasing procedure. Therefore $F(\tau,\ket\psi)=\bra{\phi_0}\tau_m\ket{\phi_0}=\alpha_0 $ and so the trace distance is bounded above by $$\|\tau-\ket\psi\bra\psi\|_1 \leq 2\sqrt{1-\alpha_0}.$$

 To bound $\alpha_0$, we use the relation between $\lambda_j$ and $\alpha_k$:
 \begin{align*}
\lambda_j&=M_\tau(2j-1,2j)=\frac{i}{2}\Tr(\tau_m[c_{2j-1}, c_{2j}])\\
&=\Tr(ic_{2j-1}c_{2j}\tau_m)=\sum_k\beta^k_j\alpha_k.
\end{align*}

Since $\alpha_k$ are probability coefficients, i.e. $\sum_k\alpha_k=1$, we have that $$\lambda_j=1-2\sum_{k:\beta_j^k=-1}\alpha_k=2\sum_{k:\beta^k_j=1}\alpha_k-1.$$

 Using the bound on the trace norm $\nu\geq\|\Lambda\tau\otimes\tau\Lambda\|_1=2m-\Tr M_\tau^TM_\tau=2m-2\sum_j\lambda_j^2,$ we find that $$\sum_j\lambda_j^2\geq m-\nu/2.$$ Since every $\lambda_j^2\leq 1$, it follows that for every $j$, $\lambda_j^2\geq 1-\nu/2.$
Therefore for every $j=1,...,m$, we have $1-\nu/2\leq\lambda^2_j\leq 1$. Consider two possible cases:
\begin{enumerate}
\item In the case when $\sqrt{1-\nu/2}\leq \lambda_j\leq 1$, we have $\sum_{\beta^k_j=-1}\alpha_k\leq \frac{1}{2}(1-\sqrt{1-\nu/2})$. Note that $\alpha_0$ is not in the sum, since in the case when $\lambda_j$ is close to $1$ and so $\beta^0_j=1$.

\item In the case when $-\sqrt{1-\nu/2}\geq \lambda_j\geq -1$, we have $\sum_{\beta^k_j=1}\alpha_k\leq \frac{1}{2}(1-\sqrt{1-\nu/2})$. Note that $\alpha_0$ is not in the sum, since in the case when $\lambda_j$ is close to $-1$ and so $\beta^0_j=-1$.
\end{enumerate}

Summing all inequalities on $\alpha$-s (for every $j$), we obtain 
$$\sum_{j=1}^m\sum_{k:\beta^k_j=-\sign\lambda_j}\alpha_k\leq \frac{m}{2}(1-\sqrt{1-\nu/2}).$$
 There is no $\alpha_0$ in the sum above and the only $\alpha$ absent in the sum is $\alpha_0$, so 
 $$1-\alpha_0=\sum_{k\neq0}\alpha_k\leq   \frac{m}{2}(1-\sqrt{1-\nu/2}),$$
therefore 
$$\|\tau-\ket\psi\bra\psi\|_1 \leq 2\sqrt{1-\alpha_0}\leq \sqrt{m\,\nu}. $$

\end{proof}

The null-space of the operator $\Lambda$ is spanned by states $\ket{\psi, \psi}$, where $\psi$ is Gaussian, and therefore is called \textit{Gaussian-symmetric subspace}. This fact was partially proved in \cite{Terhal}, Appendix A. The complete proof is provided here in Appendix A.

In \cite{Terhal}, a semi-definite program was introduced to determine whether a state is convex-Gaussian or not. A state that is not convex-Gaussian fails the program. The criterion uses the notion of extendibility. A state $\rho\in\cC_{2m}$ is said to have a \textit{$n$-Gaussian-symmetric extension} iff there exists a state $\rho^{(n)}\in\cC_{2m}^{\otimes n}$ such that 
\begin{enumerate}
\item $\rho^{(n)}\geq 0$ and $\Tr \rho^{(n)}=1$
\item $\Tr_{2,...,n}\rho^{(n)}=\rho $ \textit{(extension)}
\item $ \Lambda^{k,l}\rho^{(n)}=0, \text{ for any }k\neq l$ \textit{(symmetry)}.
\end{enumerate}
In this paper the set of states that has a $n$-Gaussian-symmetric extension is denoted by 
\begin{align*}
G^{(n)}=\{\rho\in\cC_{2m}\ : \ &\rho\text{ has }n\text{-Gaussian-symmetric}\\
&\text{extension }\rho^{(n)}\in\cC_{2m}^{\otimes n}\}.
\end{align*}

Any convex-Gaussian state $\rho=\sum_i p_i\ket{\psi_i}\bra{\psi_i}$ can be extended to any number of parties, i.e. for any $n$, the state $\rho^{(n)}=\sum_i p_i\ket{\psi_i}\bra{\psi_i}^{\otimes n}$ is the $n$-Gaussian-symmetric extension of the convex-Gaussian state $\rho\in G^{(n)}$. Therefore $G^{(n)}\supseteq G_c$, for any $n$. Clearly, $G^{(n+1)}\subseteq G^{(n)}$ for any $n$, so there is the following sequence of inclusions $G^{(1)}\supseteq G^{(2)}\supseteq...\supseteq G^{(n)}\supseteq...\supseteq G_c$. 

The following semi-definite program was introduced in \cite{Terhal}:

\textbf{Program 1.} \begin{align*} \text{\textit{Input:}} &\ \ \rho\in\cC_{2m}\text{ and an integer }n\geq 2\\
\text{ \textit{Body:}} &\ \text{ Is there } n\text{-Gaussian-symmetric extension }\rho^{(n)}\\
&\text{ of a state }\rho\\
\text{\textit{Output:}} &\ \text{ \textbf{yes,} then provide }\rho^{(n)},\text{ or \textbf{no}}.
\end{align*}

In \cite{Terhal}, Theorem 2, it was shown that if a state has a $n$-Gaussian-symmetric extension for all $n$, then the state is convex-Gaussian. Therefore the family of semi-definite programs forms a criterion consisting in checking if a state $\rho\in G^{(n)}$ for all $n$. In the next section we show that in finite time one can get arbitrary close to the set of convex-Gaussian states $G_c$ from the outside, i.e. we prove that $G^{(n)}$ converge to $G_c$ from outside, $\lim_{n\rightarrow\infty} G^{(n)}=G_c$.  This proves that the criterion is complete. In other words, if a state $\rho\in\cC_{2m}$ is not convex-Gaussian, then there exists $n$ such that $\rho\notin G^{(n)}$, so all non-convex-Gaussian states can be eventually detected.

\section{Quantitative bound}

How close is a an $n$-extendible state to a set of convex-Gaussian states $G_c$? We construct a convex-Gaussian state that is close to a given $n$-extendible state and show that the distance between them converges to zero as $n$ goes to infinity.

\begin{theorem}\label{thm:outside}
Suppose that an even state $\rho\in\cC_{2m}$ has an extension to $n$ number of parties, with $n\geq 1$. Then the state $\rho$ is close to a convex-Gaussian state, i.e. there exists a convex-Gaussian state $\sigma(n)=\sum_jp_j\ket{\psi_j(n)}\bra{\psi_j(n)}$ such that
$$ \|\rho-\sigma(n)\|_1\leq {\epsilon}(n) :=\min\{2,\, 10\|\Lambda\|_1^2\frac{2^{2m}}{n^{1/3}}\}.$$
Here  $\|\Lambda\|_1=2(m+1)\binom{2m}{m+1}$.
\end{theorem}
\begin{proof}
For any state $\sigma(n)$ the upper bound $ \|\rho-\sigma(n)\|_1\leq2$ follows from the linearity of the trace norm and the fact that the trace norm of any state is one.

To prove the non-trivial upper bound, we assume that $n$ is sufficiently large. By definition, the state $\rho$ has an extension $\rho^{(n)}$ to $n$ parties if $\Tr_{2,...n}\rho^{(n)}=\rho$, $\Tr\rho^{(n)}=1$, $\rho^{(n)}\geq 0$ and for every pair $1\leq k\neq l\leq n$, $\Lambda^{k,l}\rho^{(n)}=0$. 

Since the state $\rho^{(n)}$ is in the null-space of every operator $\Lambda_{k,l}$, its restriction to every two spaces $\rho^{(n)}_{kl}$, $1\leq k\neq l\leq n$, is symmetric (see Theorem \ref{thm:null} in Appendix A). And so the state $\rho^{(n)}$ is symmetric on $\cC_{2m}^{\otimes n}$. Therefore we may invoke quantum de Finetti theorem (see e.g., \cite{Christandl} Theorem II.8). For large enough $n$, there exist states $\tau_j(n)\in\cC_{2m}$ and probability coefficients $p_j(n)\geq 0$, $\sum_jp_j(n)=1$, such that
$$\|\rho^{(n)}_{1,2}-\sum_jp_j(n)\tau_j(n)\otimes\tau_j(n)\|_1\leq \gamma(n):=4\frac{2^m}{n},$$
where $\rho^n_{1,2}=\Tr_{3,...,n}\rho^{(n)}\in\cC_{2m}\otimes \cC_{2m}$ is a partial trace of the extension.

Applying the operator $\Lambda$ to the state inside the trace norm above, we obtain
\begin{align*}
&\|\Lambda\|_1^2\, \gamma(n)=:\Delta(n)\geq\|\Lambda\sum_jp_j(n)\tau_j(n)\otimes\tau_j(n)\Lambda\|_1\\
&=\sum_jp_j(n)\|\Lambda(\tau_j(n)\otimes\tau_j(n))\Lambda\|_1.
\end{align*}
The right-hand side equality is due to the fact that every term is positive, i.e. $p_j(n)\geq 0$ and $\Lambda(\tau_j(n)\otimes\tau_j(n))\Lambda$ is a positive operator. 

Therefore, for every $j=1,...m$, we obtain $$p_j(n)\|\Lambda(\tau_j(n)\otimes\tau_j(n))\Lambda\|_1\leq \Delta(n).$$

Consider two possible cases, when the coefficient $p_j(n)$ is small and big enough, i.e. for large enough $n$,
\begin{enumerate}
\item for every $j$, such that $p_j(n)<{n}^{2/3}{\Delta(n)}$, take any pure Gaussian state $\ket{\psi_j(n)}$ in the desired convex-Gaussian state $\sigma(n)=\sum_jp_j(n)\ket{\psi_j(n)}\bra{\psi_j(n)}$;
\item for every $j$, such that $p_j(n)\geq {n}^{2/3}{\Delta(n)}$, we have that $\|\Lambda(\tau_j(n)\otimes\tau_j(n))\Lambda\|_1\leq \Delta(n)/p_j\leq1/{n}^{2/3}.$ In the next lemma we prove that such a state, $\tau_j(n)$, is close to a pure Gaussian state.
\end{enumerate}

Note that here one can take a critical value of $p_j$ in the form $n^\epsilon\Delta(n)$ for any $0<\epsilon<1$. We chose $\epsilon=2/3$ for later convenience.

Applying Lemma \ref{lemma:app} with $\nu_j(n)=1/{n}^{2/3}$ to our situation, we obtain that for every $j$, such that $p_j(n)\geq n^{2/3}\Delta(n)$, the state $\tau_j(n)$ is close to a pure Gaussian state $\ket{\psi_j(n)}$ constructed in the Lemma. As it was said before, for every $j$, such that $p_j(n)<n^{2/3}\Delta(n)$, we set $\ket{\psi_j(n)}$ to be any Gaussian state. Then the state $\rho$ is close to the convex-Gaussian state:
\begin{align*}
\|\rho&-\sum_jp_j\ket{\psi_j(n)}\bra{\psi_j(n)} \|_1\\
\leq& \|\rho-\sum_jp_j\tau_j(n) \|_1+\sum_j p_j\|\tau_j(n)-\ket{\psi_j(n)}\bra{\psi_j(n)}\|_1\\
\leq& \gamma(n)+\sum_{j: p_j\geq {n}^{2/3}\Delta(n)} p_j\|\tau_j(n)-\ket{\psi_j(n)}\bra{\psi_j(n)}\|_1\\
&+2\sum_{j: p_j< {n}^{2/3}\Delta(n)} p_j\\
\leq& \gamma(n)+\sqrt{m\,\nu(n)}+2{n}^{2/3}\Delta(n)2^m\\
<& 10\|\Lambda\|_1^2\frac{2^{2m}}{n^{1/3}}.
\end{align*}
\end{proof}

From this theorem we see that the sets of $n$-extendible states $G^{(n)}$ are approximating the set of convex-Gaussian states $G_c$ from the outside. In the next section we develop a criterion characterizing the set of convex-Gaussian from the inside.

 \section{Criterion for convex-Gaussian states detection}
 
 Remind the Proposition \ref{Convex}, which says that a small enough perturbation of the normalized identity state by any state is contained in the set of convex-Gaussian states. Since with growing $n$, the $n$-extendible states are getting closer to the set of convex-Gaussian states, it is reasonable to believe that one may take a bigger part of this state in the perturbation and still stay in the set of convex-Gaussian states. In other words, Proposition \ref{Convex} claims that for any state $\rho\in\cC_{2m}$ there exists $\delta>0$ such that the state
 $$\tilde{\rho}:=(1-\delta)\rho+\delta\, I/2^m $$ is convex-Gaussian. We claim that for states $\rho\in G^{(n)}$ that have extension to $n$ number of parties value $\delta$ goes to $0$ as $n$ goes to infinity. Here we switched the notation from $\epsilon$ in Proposition \ref{Convex} to $\delta=1-\epsilon$, to emphasize that $\epsilon$ is growing for $n$-extendible states with growing $n$.
 
For any $n\geq 1$,  define sets
 $$\tilde{G}^{(n)}=\{\tilde{\rho}:=(1-\delta(n))\rho+\delta(n)I/2^m\, |\, \rho\in G^{(n)} \},$$
 where 
 \begin{equation}\label{delta}
 \delta(n):=\frac{\chi(n)}{1+\chi(n)} \ \text{ and }\ \chi(n):=\frac{\epsilon(n)}{2}\frac{(2m)!}{m!}, 
 \end{equation} here $\epsilon(n)$ is the distance from a complement of the set $G^{(n)}$ to the set of convex-Gaussian states $G_c$, which can be taken from Theorem \ref{thm:outside}, $\epsilon(n)=\min\{2,\, 10\|\Lambda\|_1^2\frac{2^{2m}}{n^{1/3}}\}.$ Note that $\delta(n)\rightarrow 0$ as $n\rightarrow \infty$.

In the next theorem we prove that any state $\tilde{\rho}\in\tilde{G}^{(n)}$ is convex-Gaussian, i.e. $\tilde{G}^{(n)}\subseteq G_c$. From the proof it will also be clear that we can perturb any convex-Gaussian state to a state inside $\tilde{G}^{(n)}$.
 
 \begin{theorem}\label{thm:inside}
 For any $n$-extendible state $\rho\in G^{(n)}$, $n\geq 1$, the state $\tilde{\rho}=(1-\delta(n))\rho+\delta(n)I/2^m$ is convex-Gaussian for the above choice of $\delta(n)$ (\ref{delta}).  \\
 Moreover, for any convex-Gaussian state $\rho_c$ and for any $n$, there exists a state $\omega_n\in\tilde{G}^{(n)}$ such that they are close, $$\|\rho_c-\omega_n\|_1\leq 2\delta(n).$$
  In particular, the sets $\tilde{G}^{(n)}$ approximate the set of convex-Gaussian states $G_c$ from inside.

 \end{theorem}
 \begin{proof}
 From Theorem \ref{thm:outside} we know that for every $n$-extandible state $\rho\in G^{(n)}$ there exists a convex-Gaussian state $\sigma_n(\rho)\in G_c$ such that they are close enough $\|\rho-\sigma_n(\rho) \|_1\leq \epsilon(n)$, with $\epsilon(n)\rightarrow 0$ as $n\rightarrow \infty$. 
 
 Let $\tau_n:=\rho-\sigma_n(\rho)$. Then $\Tr\, \tau_n=0$ and the trace norm is small $\|\tau_n\|_1\leq \epsilon(n)$. Write $\tau_n$ as a linear combination of even correlators and identity as in (\ref{even_op}), i.e. $$\tau_n=\alpha_0(n) I/2^m+\sum_S\alpha_S(n) C_S,$$ where $S=\{a_1,..,a_{2k}\}\subset\{1,...,2m\}$ is a subset of even number of elements, $\alpha_S(n)\in\bR$ and $C_S=i^k c_{a_1}...c_{a_{2m}}$.
 
 Since $\Tr\, \tau_n=0$, we have $\alpha_0(n)=0$.
Using  (\ref{trace_norm}), the trace norm can be written as 
$$\|\tau_n\|_1=\max_{-I\leq \Pi\leq I}\Tr \{\Pi\, \tau_n\}.$$ 
For any $S$, take $\Pi_S:=\sign(\alpha_S(n)) C_S$. Then we obtain 
$$\epsilon(n)\geq \|\tau_n\|_1\geq\Tr(\Pi_S\, \tau_n)=|\alpha_S(n)|\Tr(I)=|\alpha_S(n)|2^m.$$
Therefore, for every subset $S$ the coefficient is small $|\alpha_S(n)|\leq \epsilon(n)/2^m$.
 
 From the proof of Proposition \ref{Convex} (see \cite{Terhal}) we may write $$\tau_n=\sum_{a}\omega_{a}(n)\zeta (a)-c(n)\, I/2^m,$$ 
 where $\{\zeta(a)\}$ is an over-complete basis of $(2m)!/m!$ pure Gaussian states, $\omega_{a}(n)\geq 0$ and the constant $c(n)=\sum_a\omega_a(n)=\sum_S\sum_{a}|\alpha_S(n)|$. Taking into account that every coefficient $|\alpha_S(n)|$ is small, the constant $c$ can be bounded from above by $c(n)\leq \frac{\epsilon(n)}{2}\frac{(2m)!}{m!}=\chi(n)$. 
 
The state $\tilde{\rho}=(1-\delta(n))\rho+\delta(n)I/2^m$ can be then written as 
 \begin{align*}
 \tilde{\rho}=&(1-\delta(n))\sigma_n(\rho)+(1-\delta(n))\sum_{a}\omega_{a}\zeta (a)\\
 & +\Bigl(\delta(n) - (1-\delta(n))c\Bigr)I/2^m.
 \end{align*}
 From definition $\delta(n)={\chi(n)}/({1+\chi(n)})$ and upper bound $c(n)\leq \chi(n)$, the coefficient in front of the identity in the equation above is non-negative, i.e. $\delta(n) - (1-\delta(n))c\geq 0$. Since the normalized identity, $\sigma_n(\rho)$ and $\zeta_a$ are convex-Gaussian states, the state $\tilde{\rho}$ is convex-Gaussian.
 
 Let $\rho_c\in G_c$ be a convex-Gaussian state. For any $n\geq 1$ this state is also $n$-extendible, $\rho_c\in G^{(n)}$. Therefore, applying the result we just proved, $\omega_n=(1-\delta(n))\rho_c+\delta(n)I/2^m\in\tilde{G}^{(n)}$ and the distance is bounded
 $$\|\rho_c-\omega_n\|_1=\delta(n)\|\rho_c-I/2^m\|_1\leq 2\delta(n). $$ 
 
 \end{proof}
 
Since every convex-Gaussian state can be approximated with an arbitrary precision by states in  sets $\{\tilde{G}^{(n)}\}_{n\geq 1}$, Theorem \ref{thm:inside} gives rise to a complete criterion determining whether a state is convex-Gaussian or not. The criteria comes in the form of a hierarchy of the following semi-definite programs: for $\chi(n)=\frac{\epsilon(n)}{2}\frac{(2m)!}{m!}$ and $\epsilon(n)=\min\{2,\, 10\|\Lambda\|_1^2\frac{2^{2m}}{n^{1/3}}\}$
 
 \textbf{Program 2.} \begin{align*} \text{\textit{Input:}}& \ \ \rho\in\cC_{2m}\text{ and }n\geq 1,\\
\text{ \textit{Body:}}& \text{ Is there an } n\text{-extendible state }\sigma\in G^{(n)}\text{ s.t.}\\
&\ \ \rho=(1-\delta(n))\sigma+\delta(n)\, I/2^m, \\
&\text{ for  } \frac{\chi(n)}{1+\chi(n)}\leq\delta(n)\leq 1\\
\text{\textit{Output:}}& \ \text{ \textbf{yes} or \textbf{no}}.
\end{align*}

Note that $\epsilon(n)$ here is the distance from the compliment of a set $G^{(n)}$ to the set of convex-Gaussian states $G_c$. Our bound on $\epsilon(n)$ in Theorem \ref{thm:outside} is non-trivial only for $n$ of order higher than $2^{18m}$. Refining the bound on a distance between any $n$-extendible state and $G_c$ for small $n$ will improve the program for smaller $n$. Nevertheless, the above program is valid for any $n$ starting from 1.

Together with Program 1, one can alternatively run both programs for every $n$. If a state fails Program 1 for some $n$, the state is not convex-Gaussian, and if a state passes Program 2 for some n, it is convex-Gaussian.

\section{Conclusion}

We have introduced a complete criterion that is able to detect whether a state is a convex mixture of Gaussian states. This criterion is complementary to the one introduced in \cite{Terhal}, which we proved here is indeed a complete criterion approximating the set of convex-Gaussian states $G_c$ from the outside. Our complementary criterion approximates the set $G_c$ from the inside by perturbing $n$-extendible states to make them convex-Gaussian. The semi-definite program depends on the distance from the compliment of a set of $n$-extendible states to the set of the convex-Gaussian states, and, unfortunately, we were able to provide a non-trivial bound only for large $n$. Further work is needed to improve the bound for small $n$. Nevertheless, a complementary criterion gives a good taste on how both tests can work together faster to determine whether a state is convex-Gaussian or not, as well as it shows a way to characterize the set of convex-Gaussian states from the inside.\\
\\

\textbf{Acknowledgment}
A.V. is grateful to Barbara M. Terhal for vital insight into the problems discussed in the paper and a valuable feedback on the manuscript.

A.V. acknowledges funding through the European Union via QALGO FET-Proactive Project No. 600700.

\section{Appendix A. Gaussian-Symmetric subspace}

Here we provide a complete proof that the null-space of the operator $\Lambda$ is spanned by states $\ket{\psi, \psi}$, where $\psi$ is Gaussian. We repeat the partial proof given in \cite{Terhal} and complete it by proving Lemma \ref{Invariant_space}.

Define a 'FLO twirl' as the map 
\begin{equation}\label{Twirl}
\cS(\rho)=\int_{FLO} dU\  U\otimes U \, \rho\, U^\dagger\otimes U^\dagger
\end{equation} 
for any $\rho\in\cC_{2m}\otimes\cC_{2m}$. Here $U\in FLO$ acts as follows: 
$$Uc_jU^\dagger=(Rc)_j=\sum_i R_{ji}c_i, \text{ for every } j=1,...,2m, $$
with $R\in SO(2m)$. The integral in (\ref{Twirl}) $\int_{FLO}dU$ is defined by taking the Haar measure over the real orthogonal matrices $R$ induced by $U$. The map $\cS(\rho)$ is normalized such that it is trace-preserving. First we prove a statement about an invariant subspace that was made in \cite{Terhal}.

\begin{lemma}\label{Invariant_space}
The invariant subspace of $\cS$ is spanned by $I\otimes I, \Lambda,..., \Lambda^{2m}$ operators.
\end{lemma}
\begin{proof}

Let us denote $c_0=I$, then modify the every matrix $R$ to become $\tilde{R}=I\oplus R$, where $I$ is $1$-dimensional matrix.  The action of $U$ stays the same for the modified matrix $R$:
$$\tilde{U}c_j\tilde{U}^\dagger=\sum_i \tilde{R}_{ji}c_i, \text{ for every } j=0,1,...,2m,  $$ 
  
Any state $\rho\in\cC_{2m}\otimes\cC_{2m}$ can be written as follows:
\begin{equation}\label{rho}
\rho=\sum_{J,L}\lambda_{JL}c_{j_1}...c_{j_{2m}}\otimes c_{l_1}...c_{l_{2m}},
\end{equation}
where $J,L\in\cH$, $\cH: =\{I=(i_1,...,i_{2m}): \, i_k=0,...,2m \text{ for any }1\leq k\leq 2m\}$ and $c_0=I$.

The cardinality of $\cH$ is $n:=(2m+1)^{2m}$. Order sequences in $\cH$ in the following manner:\\
$(0,...,0)$\\
$(0,...,0,l)$ where $1\leq l\leq 2m$\\
...\\
$(l,0,...,0)$\\
$(0,...,0,l_1,l_2)$ with $1\leq l_1,l_2\leq 2m $
...\\
...\\
$(l_1,...,l_{2m})$, with $1\leq l_1,...,l_{2m}\leq 2m$.

Form a $n\times n$ matrix $\lambda$ such that $(\lambda)_{JL}=\lambda_{JL}$. 
Also form a $n\times n$ matrix $\hat{R}$ such that $$\hat{R}_{IJ}=\tilde{R}_{i_1j_1}...\tilde{R}_{i_{2m}j_{2m}}.$$
$\hat{R}$ is an orthogonal matrix since 
\begin{align*}
&\sum_I\hat{R}_{IJ}\hat{R}_{IK}=\sum_I\tilde{R}_{i_1j_1}...\tilde{R}_{i_{2m}j_{2m}}\tilde{R}_{i_1k_1}...\tilde{R}_{i_{2m}k_{2m}}\\
&=\delta_{j_1k_1}...\delta_{j_{2m}k_{2m}}=\delta_{JK}.
\end{align*}
By construction $\hat{R}$ is a direct product of matrices: $1$-dimensional identity matrix, $2m$ matrices of the dimension $2m$, $\binom{2m}{2}$ matrices of dimension $(2m)^2$,... $\binom{2m}{k}$ matrices of dimension $(2m)^k$..., and 1 matrix of dimension $(2m)^{2m}$.

The state $\rho$ is invariant under group $U\otimes U$, therefore 
\begin{align*}
&\rho=\sum_{J,L}\lambda_{JL}c_{j_1}...c_{j_{2m}}\otimes c_{l_1}...c_{l_{2m}}\\
&=\sum_{J,L}\lambda_{JL}(\tilde{U}c_{j_1}\tilde{U}^\dagger)...(\tilde{U}c_{j_{2m}}\tilde{U}^\dagger)\otimes (\tilde{U}c_{l_1}\tilde{U}^\dagger)...(\tilde{U}c_{l_{2m}}\tilde{U}^\dagger)\\
&=\sum_{J,L}\Bigl(\sum_{IK}\lambda_{IK}\hat{R}_{IJ}\hat{R}_{KL}\Bigr)\, c_{j_1}...c_{j_{2m}}\otimes c_{l_1}...c_{l_{2m}}\\
&=\sum_{J,L}(\hat{R}^T\lambda\hat{R})_{JL}\, c_{j_1}...c_{j_{2m}}\otimes c_{l_1}...c_{l_{2m}}.
\end{align*}

Coefficients in front of the same correlators have to be the same, so $\lambda$ commutes with any orthogonal matrix $\hat{R}$. Because of the structure of $\hat{R}$, the matrix that commutes with any $\hat{R}$ has the following structure:\begin{align*}
\lambda=&a_{00}I\oplus a_{1,1}I\oplus...\oplus a_{1, 2m}I\oplus... \\
& ....\oplus a_{k,1}I\oplus...\oplus a_{k,\binom{2m}{k}}I\oplus...\, ....\oplus a_{2m,1}I,
\end{align*}
where the identity corresponding to the coefficient $a_{k,l}\in\bR$ has dimension $(2m)^k$.
Therefore the state $\rho$ has the form
\begin{align*}
\rho&=\alpha_0 I\otimes I+\alpha_1\Lambda+...+\alpha_{2m}\Lambda^{2m},
\end{align*}
for some coefficients $\alpha_j$.
In other words, the invariant subspace of $S$ is spanned by $I\otimes I, \Lambda,..., \Lambda^{2m}$ operators.
\end{proof}

Now we repeat the statement and its proof about the null-space of $\Lambda$.

\begin{theorem}\label{thm:null}
(\cite{Terhal}, Appendix A) The projector onto the null-space of $\Lambda=\sum_i c_i \otimes c_i$ is $\Pi_{\Lambda=0}={2m \choose m} {\cal S}(\ket{0,0}\bra{0,0})$ where $\ket{0}$ is a (Gaussian) vacuum state with respect to some set of annihilation operators $a_i$. Thus the states $\ket{\psi,\psi}$ where $\psi$ is a pure fermionic Gaussian state span the null-space of $\Lambda$ which implies that the null-space of $\Lambda$ is a subspace of the symmetric subspace.
\end{theorem}

\begin{proof}
In order to prove that $\Pi_{\Lambda=0}={2m \choose m} {\cal S}(\ket{0,0}\bra{0,0})$, we note that both the l.h.s. and the r.h.s. are $U \otimes U$-invariant where $U$ is any FLO transformation. Thus instead of considering whether 
$${\rm Tr} X \Pi_{\Lambda=0}={\rm Tr}(X {2m \choose m} {\cal S}(\ket{0,0}\bra{0,0}))$$
 for any $X$, we can just consider the trace with respect to invariant objects ${\cal S}(X)$.  
 
 From Lemma \ref{Invariant_space}, we know that $\Lambda^i$ for $i=0,1,2,\ldots,2m$ (and linear combinations thereof) are the only invariants under the group $U \otimes U$ where $U$ is FLO transformation. Clearly, 
 $${\rm Tr} \Lambda^i {2m \choose m} {\cal S}(\ket{0,0}\bra{0,0})=0={\rm Tr} \Lambda^i \Pi_{\Lambda=0}$$
  for all $i \neq 0$ while ${\rm Tr} \Pi_{\Lambda=0}={2m \choose m}$ fixes the overall prefactor. 
  
  Having established the form of the projector, it follows directly that the states $\ket{\psi,\psi}$ for any Gaussian $\psi$ span the null-space (Assume this is false and hence a state in the null-space $\ket{\chi}=\ket{\chi_{in}}+\ket{\chi_{out}}$ where $\ket{\chi_{in}}$ is in the span of $\ket{\psi,\psi}$ while $\ket{\chi_{out}}$ is w.l.o.g. orthogonal to any $\ket{\psi,\psi}$. We have $\Pi_{\Lambda=0}\ket{\chi}=\ket{\chi}$ while ${2m \choose m} {\cal S}(\ket{0,0}\bra{0,0}) \ket{\chi}=\ket{\chi_{in}}$ arriving at a contradiction.) 
  
  As $\ket{\psi,\psi}$ for Gaussian pure states $\psi$ span the null-space, and $P \ket{\psi,\psi}=\ket{\psi,\psi}$ with $P$ the SWAP operator, the null-space is a subspace of the symmetric subspace.
\end{proof}

\section{Appendix B. Gaussian-symmetric extension}

In \cite{Bravyi} it was shown that there is an isomorphism between $\cC_{2m}\otimes \cC_{2m}$ and $\cC_{4m}$. Here we provide an explicit isomorphism between $\cC_{2m}^{\otimes n}$ and $\cC_{2mn}$ and explore how one could alternatively express semi-definite Program 1 as an extension of $\rho$ to a physical system with $2mn$ Majorana fermions.

The original Program 1 was the following: given a state $\rho\in\cC_{2m}$ find an extension $\rho_{ext}\in \cC_{2m}^{\otimes n}$ such that 
\begin{enumerate}
\item $\Tr_{2,...,n}\rho_{ext}=\rho $ 
\item $ \Lambda^{k,l}\rho_{ext}=0, \text{ for any }k\neq l.$
\end{enumerate}

The present scheme is the following: given a state $\rho\in\cC_{2m}$ find "an extension" $\mu\in\cC_{2mn}$ that satisfies two condition which we will determine below.

The isomorphism $J$ between $\cC_{2mn}$ and $\cC_{2m}^{\otimes n}$ is given by
$$J(c_j)=I\otimes ...\otimes c_j\otimes P\otimes...\otimes P, \text{ for } 2m(k-1)+1\leq j\leq 2mk,$$
where $c_j$ stands on the $k$-th component and $P=i^mc_1...c_{2m}\in\cC_{2m}$ is a parity operator, $P^2=I$ and $P^\dagger=P$.

Since $J$ preserves commutation relations between the generators, it extends to arbitrary operators by linearity and multiplicativity, $J(XY)=J(X)J(Y).$ And therefore $J$ is the isomorphism.

\textit{As an example consider the case $n=2$. Any state in $\cC_{4m}$ can be written as }
\begin{equation}\label{mu}
\mu=\sum_k\alpha_k Z_k^1Z_k^2\ \in \cC_{4m},
\end{equation}
\textit{where every monomial in $\mu$ is written as a product of two operators $Z^1$ and $Z^2$ that act on $c_1,...,c_{2m}$ and $c_{2m+1},...,c_{4m}$ respectively. This state gets mapped to the following state by $J$: }
$$J(\mu)=\sum_k\alpha_k Z_k^1\otimes P^{\epsilon_k}Z^2_k\ \in\cC_{2m}\otimes \cC_{2m}, $$
\textit{where $\epsilon_k=0$, if $Z^1_k$ is even and $\epsilon_k=1$, if $Z_k^1$ is odd.}
\textit{Note that the reverse can also be easily obtained. Any state $$\sigma=\sum_k \beta_k Z_k^1\otimes Z_k^2 \ \in\cC_{2m}\otimes \cC_{2m} $$ gets mapped onto a state $$J^{-1}(\sigma)=\sum_k \beta_K Z_k^1P^{\epsilon_k}Z_k^2 \ \in\cC_{4m},$$ where $P=c_{2m+1}...c_{4m}.$}

\textit{The operator $\Lambda=\sum_{j=1}^{2m}c_j\otimes c_j$ on $\cC_{2m}\otimes\cC_{2m}$ is equivalent to the operator} $$\Gamma=\sum_{j=1}^{2m}(-1)^{m-j}c_jc_{2m+1}...\hat{c}_{2m+j}...c_{4m} \text{ \ on } \cC_{4m}, $$
\textit{where $\hat{c}_k$ denotes the omission of the operator $c_k$.}

\textit{Therefore the condition $\Lambda(J(\mu))=0$ is equivalent to $\Gamma(\mu)=0$ for any $\mu\in\cC_{4m}$.}

Operator $\Gamma$ can be defined using any generators from $\cC_{2mn}$: for any $1\leq k\neq l\leq n$,
\begin{align}
&\Gamma^{k,l}=\label{Gamma}\\
&\sum_{j=1}^{2m}(-1)^{m-j} \ c_{2m(k-1)+j}\ \ c_{2m(l-1)+1}...\hat{c}_{2m(l-1)+j}...c_{2ml}.\nonumber
\end{align}

Therefore condition $2.$ above is equivalent to the following one 
\begin{equation}\label{second_cond}
\Gamma^{k,l}\mu=0, \text{ for any }k\neq l,
\end{equation}
where $\Gamma$ is defined in (\ref{Gamma}) and $\mu\in\cC_{2mn}$.

\textit{Let us return to the case $n=2$. In this case if the extension $\mu$ of $\rho$ looks like (\ref{mu}), the first condition to the extension gets the form} $$\rho=\Tr_2J(\mu)=\sum_j\alpha_jZ_j^1 \, \Tr (P^{\epsilon_j}Z_j^2). $$
\textit{Here $Z_j^2$ is a monomial in $c_{2m+1},... , c_{4m}$. The trace is non-zero in the only case when $P^{\epsilon_j}Z_j^2=I$. This happens in two cases}
\begin{itemize}
\item $Z_j^1$ \textit{even, then }$Z_j^2=I$,

\item $Z_j^1$ \textit{odd, then} $Z_j^2=P$.
\end{itemize}
\textit{Therefore the state $\rho$ can be obtained from its extension $\mu=\sum_j\alpha_j Z_j^1Z_j^2$ as}
\begin{align*}
\rho=2^{m}\{\sum\alpha_jZ_j^1: &\ j \text{ such that if }Z_j^1 \text{ is even and }\\
& Z_j^2=I \text{ \textbf{or} if } Z_j^1 \text{ is odd and } Z_j^2=P\}.
\end{align*}

In the general case when the extension $\mu=\sum_j \alpha_jZ_j^1...Z_j^n$ is on $\cC_{2mn}$, we obtain
\begin{align*}
\rho=\sum_j\alpha_jZ^1_j &\Tr(P^{\epsilon_j^1}Z_j^2)...\Tr(P^{\epsilon_j^{n-1}}Z_j^n)\\
=2^{m(n-1)}\Bigl\{&\sum\alpha_jZ_j^1: j\text{ such that if }Z_j^1 \text{ is even and }\\
&Z_j^2=I,...Z_j^n=I \text{ \textbf{or} if }Z_j^1\text{ is odd and }\\
&Z_j^2=P,...Z_j^n=P\Bigr\}.
\end{align*}

Thus the extension $\rho_{ext}\in\cC_{2mn}$ of a state $\rho=\sum_j \alpha_j Z_j^1\ \in\cC_{2m}$ is such that
\begin{enumerate}
\item the extension is of the form 
\begin{align*}
\rho_{ext}=&2^{-m(n-1)} \Bigl(\sum_{j : Z_j^1\text{ is even}}\alpha_j Z_j^1+\sum_{j : Z_j^1\text{ is odd}}\alpha_j Z_j^1P...P \Bigr)\\
&+Corr,
\end{align*} 
where the correlations $Corr=\sum_j\beta_jZ_j^1...Z_j^n$ are such that there are no non-zero terms such that $Z_j^1$ is even and $Z_j^2=I$,...$Z_j^n=I$ and also there are no non-zero terms such that $Z_j^1$ is odd and $Z_j^2=P$,...$Z_j^n=P$. 

\item $\Gamma^{k,l}\rho_{ext}=0, \text{ for any }1\leq k\neq l \leq n$, where $\Gamma^{k,l}$ is defined by (\ref{Gamma}).
\end{enumerate}

If we take an {even} state $\rho=\sum_j \alpha_j Z_j^1\ \in\cC_{2m}$ the extension $\rho_{ext}\ \in\cC_{2mn}$ is a state, satisfying condition 2. above, of the following form
$$\rho_{ext}=2^{-m(n-1)}\rho+Corr.$$
Here every term in the correlations $Corr=\sum_j\beta_jZ_j^1...Z_j^n$ acts nontrivially on the last $(n-1)$ spaces spanned by $c_{2m+1}\, ...\, c_{2mn}$. 

To make the extension an even state introduce
\begin{equation}\label{even_ext}
\tilde{\rho}_{ext}=\frac{1}{2}(\rho_{ext}+C_{all}\, \rho_{ext}\, C_{all}),
\end{equation}
where $C_{all}=i^{mn}c_1...c_{2mn}\in\cC_{2mn}$. Since $C_{all}$ commutes with even operators and anti-commutes with odd ones, we see that $\tilde{\rho}_{ext}$ is an even state that satisfies both conditions 1. and 2. above, making it a valid {even} extension of an even state $\rho$.

\end{document}